\documentclass{ijuc}
\usepackage[pdftex]{graphics}
\usepackage{amsmath,amssymb,amsfonts,amsthm,bm}

\newtheorem{proposition}{Proposition}

\begin{document}

\title{$q$-VFCA: $q$-state Vector-valued Fuzzy Cellular Automata}

\author{Yuki Nishida\inst{1}\email{ynishida.cyjc1901@gmail.com}
\and Sennosuke Watanabe\inst{2}
\and Akiko Fukuda\inst{3}
\and Yoshihide Watanabe\inst{4}
}

\institute{Department of Science of Environment and Mathematical Modeling, Doshisha University, 
1-3 Tatara Miyakodani, Kyotanabe-shi, Kyoto 610-0394, Japan
\and
Department of General Education, National Institute of Technology, Oyama College, 
771 Nakakuki, Oyama-shi, Tochigi 323-0806, Japan
\and
Department of Mathematical Sciences, Shibaura Institute of Technology, 
307 Fukasaku, Minuma-ku, Saitama-shi, Saitama 337-8570, Japan
\and
Department of Mathematical Sciences, Doshisha University, 
1-3 Tatara Miyakodani, Kyotanabe-shi, Kyoto 610-0394, Japan
}

\def\received{Received...}

\maketitle

\begin{abstract}
Elementary fuzzy Cellular Automata (CA) are known as continuous counterpart of elementary CA,
which are 2-state CA, via the polynomial representation of local rules.
In this paper, we first develop a new fuzzification methodology for $q$-state CA.
It is based on the vector representation of $q$-state CA, that is,
the $q$-states are assigned to the standard basis vectors of the $q$-dimensional real space
and the local rule can be expressed by a tuple of $q$ polynomials.
Then, the $q$-state vector-valued fuzzy CA are defined by expanding the set of the states
to the convex hull of the standard basis vectors in the $q$-dimensional real space.
The vector representation of states enables us to 
enumerate the number-conserving rules of $3$-state vector-valued fuzzy CA
in a systematic way.
\end{abstract}

\keywords{cellular automata, fuzzy cellular automata, vector-valued cellular automata, 
conservation law, number-conserving rule, periodic boundary condition}

\section{Introduction}

One-dimensional Cellular Automata (CA) are linearly arranged arrays of cells that evolve simultaneously in accordance with the local update rule depending only on their neighboring cells. 
Although rules of CA are very simple, they provide 
surprisingly rich applications and knowledges \cite{Wolfram},
e.g., traffic dynamics \cite{Habel}, evacuation process \cite{Spartalis}, cryptography \cite{Karmakar}, 
project scheduling \cite{Shimura}, image processing \cite{Tourtounis}, urban planning \cite{Aburas}, and so on.
Simple examples of CA are 2-state 3-neighbor CA, 
called Elementary CA (ECA).
Cattaneo et al.~\cite{Cattaneo} apply the fuzzification to 
the boolean operators in the disjunctive normal form of 
the local rule of ECA
and obtain a kind of continuous CA 
whose cells have states in $[0,1]$. 
Such CA are called elementary fuzzy CA and
their local rules are polynomials of the states of the neighbors.
The asymptotic behaviors of elementary fuzzy CA are studied for 
 cases with a single seed in a zero background~\cite{Mingarelli2006a}
 and for periodic cases~\cite{Betel}.
There are more detailed studies on dynamics for some specific rules, e.g., 
rule 90~\cite{Yacoubi2008,Yacoubi2011a,Flocchini2000}, rule 110~\cite{Mingarelli2003}, rule 184~\cite{Mingarelli2006b}.
The algebraic approach for elementary fuzzy CA using invariant theory is presented in~\cite{Yacoubi2011b}. 
However, to the best of our knowledge, $q$-state fuzzy CA have not been successfully formulated yet. 
For example, if we try to fuzzify 3-state CA, 
we can simply come up with expressing local rules by polynomials along a similar way to elementary fuzzy CA and 
expanding the states from $\{0,1,2\}$ to the continuous values in $[0,2]$. 
Although local rules of elementary fuzzy CA are closed in $[0,1]$, 
local rules of 3-state fuzzy CA constructed above are not closed in $[0,2]$. 
Moreover, the state ``1''  is the middle value of ``0'' and ``2'', which means that a state can be expressed by other states. 
\par
In this paper, we develop a new expression of $q$-state $n$-neighbor fuzzy CA.
To treat the $q$ states independently, 
we represent the $q$ states 
by the standard basis vectors of $\mathbb{R}^q$ instead of
$\{0,1, \dots, q-1\}$. 
Then, the states of the corresponding fuzzy CA are
in the $(q-1)$-simplex whose vertices are 
the standard basis vectors of $\mathbb{R}^q$.
In the case with $q=3$, states of the 3-state fuzzy CA belong to the interior or the boundary of the regular triangle with vertices $(1,0,0)^{\top}$, $(0,1,0)^{\top}$,
$(0,0,1)^{\top}$. 
We call such CA $q$-state Vector-valued Fuzzy CA ($q$-VFCA). 
The vector representation of CA reminds us the quantum CA~\cite{Grossing}. 
They are also continuous valued CA, 
but they do not have the discrete counterparts.
The $q$-VFCA in this paper are 
based on the conventional CA with $q$ discrete states.
Hence, we have one-to-one correspondence 
between the $q$-VFCA and the usual $q$-state CA.
Other CA represented by three-dimensional vectors are  
considered in the studies on the image processing or the encryption of 
RGB color images \cite{Faraoun,Ioannidis}, but they do not come from 3-state CA.
We also remark that fuzzy CA in this paper does not mean CA on fuzzy sets~\cite{Mraz}
or fuzzy choice of local rules~\cite{Adamatzky}.
\par
The existence of the conserved quantities  
is one of the fundamental problem in the study of periodic CA.
The additive conserved quantities of ECA and 
the elementary reversible CA are investigated in~\cite{Hattori}.
In the case of ECA, an example of the additive quantity is the sum of the states of the cells in a period, 
which is identical to the number of cells with state ``1''.
Fuk\'s and Sullivan~\cite{Fuks} give a combinatorial characterization of 
the number-conserving CA rules and compute the number of such rules for
$q$-state $n$-neighbor CA.
This ``number-conserving'' means that the state ``$k$'' is equal to the numerical value $k$ and the sum of these numerical values is conserved.
On the other hand, $q$-state vector-valued CA we focus on enable us to enumerate the rules that conserve the number of the cells with the specific state ``$k$'' by taking the sum of the $k$th entries of all vectors as the additive quantity.
The concept of number-conserving rules can be extended from $q$-state vector-valued CA to $q$-VFCA,
as shown in~\cite{Betel2011} for elementary fuzzy CA.
\par
The present paper is organized as follows. 
In Section 2, we give definitions of ECA and elementary fuzzy CA.
In Section 3, $q$-state $n$-neighbor vector-valued fuzzy CA are introduced.
Although the method presented in this paper can be applied to $q$-state $n$-neighbor vector-valued fuzzy CA,
we describe $3$-state $3$-neighbor vector-valued fuzzy CA for the convenience of the notations.
In Section 4, we enumerate all the number-conserving rules of $3$-state $3$-neighbor vector-valued fuzzy CA,
where similar computation is applicable to any other additive conserved quantities.

\section{Elementary CA and elementary fuzzy CA}

In this section, definitions and notations of ECA and elementary fuzzy CA are given. 
\par
We denote the set of the states by $Q$. 
The neighboring cells of the cell $i \in \mathbb{Z}$ are given by the set
$N(i) = \{ i - n_{\ell}, i - n_{\ell} + 1, \dots, i, \dots, i + n_r \}$.
CA are called $q$-state $n$-neighbor CA if $|Q|=q$ and $|N(i)|=n$.
Let $x_{i}^{t}\in Q$ denote the state of the cell $i$ at the time $t\in\mathbb{Z}_{\geq0}$.
The local rule $f: Q^n \to Q$ determines the evolutions of the cells by
\begin{align*}
x_i^{t+1} = f(x_{i-n_{\ell}}^t, \dots, x_i^t, \dots, x_{i+n_{r}}^t), \quad i \in \mathbb{Z},\quad  t \in\mathbb{Z}_{\geq0}.
\end{align*}
\par
We describe elementary fuzzy CA introduced in~\cite{Cattaneo}.
Let us consider the ECA rule defined by Table \ref{tab-eca}.
The local rule $h: \{0,1\}^3 \to \{0,1\}$ is expressed by the polynomial
obtained from the disjunctive normal form in boolean operations:
\begin{equation} \label{eq-efca}
\begin{aligned}
h(x,y,z) =&b_{111}xyz + b_{110} xy(1-z) + b_{101} x(1-y)z \\
&+ b_{100} x(1-y)(1-z) +b_{011} (1-x)yz \\
&+ b_{010} (1-x)y(1-z) + b_{001} (1-x)(1-y)z \\
&+ b_{000} (1-x)(1-y)(1-z).
\end{aligned}
\end{equation}
For example, the polynomial expression of the local rule of
 ECA rule 184 given in Table \ref{tab-184} is
\begin{align*}
h(x,y,z) &=xyz+x(1-y)z + x(1-y)(1-z) + (1-x)yz\\
&=x -xy + yz.
\end{align*}
If we expand the domain of $h$ to $[0,1]^{3}$ via the polynomial \eqref{eq-efca},
then we can easily check that the obtained function, denoted by $f$, 
satisfies $f([0,1]^{3}) \subset [0,1]$.
The CA defined by the local rule $f: [0,1]^3 \to [0,1]$
are called elementary fuzzy CA.
\begin{table}[t]
\begin{center}
\begin{tabular}{|c||c|c|c|c|c|c|c|c|}\hline
$xyz$ & $111$ & $110$ & $101$ & $100$ & $011$ & $010$ & $001$ & $000$ \\ \hline
$h(x,y,z)$ & $b_{111}$ & $b_{110}$ & $b_{101}$ & $b_{100}$
             & $b_{011}$ & $b_{010}$ & $b_{001}$ & $b_{000}$
\\ \hline
\end{tabular}
\end{center}
\caption{The local rule $h$ of ECA.} \label{tab-eca}
\end{table}
\begin{table}[t]
\begin{center}
\begin{tabular}{|c||c|c|c|c|c|c|c|c|}\hline
$xyz$ & $111$ & $110$ & $101$ & $100$ & $011$ & $010$ & $001$ & $000$ \\ \hline
$h(x,y,z)$ & $1$ & $0$ & $1$ & $1$ & $1$ & $0$ & $0$ & $0$
\\ \hline
\end{tabular}
\end{center}
\caption{ECA rule 184.} \label{tab-184}
\end{table}
\section{$\lowercase{q}$-state CA and $\lowercase{q}$-state fuzzy CA}
We next consider the fuzzification of $q$-state CA.
The states ``0'',``1'', $\dots$, ``$q-1$'' in the conventional CA are just symbols, not numerical values, that is, 
the states $0, 1, 2, \dots$ can be replaced with $A, B, C, \dots$, for example. 
However, if we set the state to the continuous value in $[0,q-1]$ in the fuzzification process, 
then each state can not be regarded as an independent state. 
So we consider the new formulation that the states $0,1,\dots,q-1$ are completely  independent states.
\par
A simple way to express $q$ independent states
is assigning them to the standard basis vectors 
$\bm{e}_1, \bm{e}_2, \dots, \bm{e}_q$ of $\mathbb{R}^q$.
We call the CA with $Q=\{\bm{e}_1, \bm{e}_2, \dots, \bm{e}_q\}$
$q$-state Vector-valued CA ($q$-VCA).
In this section, we construct $q$-state Vector-valued Fuzzy CA ($q$-VFCA).
To avoid the use of complicated indices,
we explain the case with $q=3$ and $n=3$,
although we can easily construct $q$-VFCA with any $q$ or $n$ in the same way.
\par
As in the case of ECA, we use the polynomial expression of the local rule for
the fuzzification of $3$-VCA.
The local rule $h:\{ \bm{e}_1,\bm{e}_2,\bm{e}_3 \}^{3} \to \{ \bm{e}_1,\bm{e}_2,\bm{e}_3 \}$ of 
$3$-VCA is expressed by the tuple of homogeneous polynomials of degree 3 as
\begin{align}
h(\bm{x},\bm{y},\bm{z}) &= \sum_{j,k,{\ell} = 1}^3 x_jy_kz_{\ell} h(\bm{e}_j,\bm{e}_k,\bm{e}_{\ell})\notag \\
&=\left(
\begin{array}{c}
\displaystyle\sum_{h(\bm{e}_j,\bm{e}_k,\bm{e}_{\ell}) = \bm{e}_1} x_j y_k z_{\ell}\\
\displaystyle\sum_{h(\bm{e}_j,\bm{e}_k,\bm{e}_{\ell}) = \bm{e}_2} x_j y_k z_{\ell}\\
\displaystyle\sum_{h(\bm{e}_j,\bm{e}_k,\bm{e}_{\ell}) = \bm{e}_3} x_j y_k z_{\ell}
\end{array}
 \right), \label{eq-3vca}
\end{align}
where $\bm{x} = (x_1,x_2,x_3)^{\top}, \bm{y} = (y_1,y_2,y_3)^{\top}, \bm{z} = (z_1,z_2,z_3)^{\top}$.
Indeed, each monomial $x_jy_kz_{\ell}$ vanishes unless $\bm{x} = \bm{e}_j, \bm{y} = \bm{e}_k$ and $\bm{z} = \bm{e}_{\ell}$.
Conversely, each map $h: \{\bm{e}_1,\bm{e}_2,\bm{e}_3\}^{3} \to \{\bm{e}_1,\bm{e}_2,\bm{e}_3\}$ of the form
\begin{align*}
h(\bm{x},\bm{y},\bm{z}) = \left( 
\begin{array}{c}
\displaystyle\sum_{j,k,{\ell}=1}^3 a_{jk{\ell}}x_jy_kz_{\ell}\\
\displaystyle\sum_{j,k,{\ell}=1}^3 b_{jk{\ell}}x_jy_kz_{\ell}\\
\displaystyle\sum_{j,k,{\ell}=1}^3 c_{jk{\ell}}x_jy_kz_{\ell}
\end{array}
\right),
\end{align*}
where $(a_{jk{\ell}},b_{jk{\ell}},c_{jk{\ell}})^{\top} \in \{\bm{e}_1,\bm{e}_2,\bm{e}_3\}$,
is the local rule of some $3$-VCA.
\par
Let $\Delta$ be the triangle in the three-dimensional space whose vertices are $\bm{e}_1, \bm{e}_2$ and $\bm{e}_3$, i.e.,
\begin{align*}
\Delta = \{ (x_1,x_2,x_3)^{\top} \ |\ x_1+x_2+x_3 = 1, x_i \geq 0, i=1,2,3 \}.
\end{align*}
Now we expand the domain of $h$ to $\Delta^{3}$ via \eqref{eq-3vca}.
In the following proposition, we prove that the obtained function, denoted by $f$, 
satisfies $f(\Delta^{3}) \subset \Delta$.
\begin{proposition}
For any $\bm{x}, \bm{y}, \bm{z} \in \Delta$, the summation of the entries of the vector
$f(\bm{x}, \bm{y}, \bm{z})$ is equal to 1, namely, 
\begin{align}
 \sum_{h(\bm{e}_j,\bm{e}_k,\bm{e}_{\ell}) = \bm{e}_1} x_j y_k z_{\ell}
+ \sum_{h(\bm{e}_j,\bm{e}_k,\bm{e}_{\ell}) = \bm{e}_2} x_j y_k z_{\ell}
+ \sum_{h(\bm{e}_j,\bm{e}_k,\bm{e}_{\ell}) = \bm{e}_3} x_j y_k z_{\ell}
=1. \label{sum}
\end{align}
\end{proposition}
\begin{proof}\normalfont
Since each monomial $x_jy_kz_{\ell}$ appears exactly once in the summands of the left-hand side of \eqref{sum}, we have
\begin{align*}
&\sum_{h(\bm{e}_j,\bm{e}_k,\bm{e}_{\ell}) = \bm{e}_1} x_j y_k z_{\ell}
+ \sum_{h(\bm{e}_j,\bm{e}_k,\bm{e}_{\ell}) = \bm{e}_2} x_j y_k z_{\ell}
+ \sum_{h(\bm{e}_j,\bm{e}_k,\bm{e}_{\ell}) = \bm{e}_3} x_j y_k z_{\ell} \\
&\qquad = (x_1 + x_2 + x_3)(y_1 + y_2 + y_3)(z_1+z_2+z_3) \\
&\qquad =  1.
\end{align*}
The last equality follows from $\bm{x}, \bm{y}, \bm{z} \in \Delta$.
\end{proof}
$3$-state Vector-valued Fuzzy CA ($3$-VFCA) are continuous CA 
whose state set is $Q=\Delta$ and
local rule is given by $f$.
\par
To visualize the evolution of $3$-VFCA, we use the RGB color system.
The states $\bm{e}_1$, $\bm{e}_2$ and $\bm{e}_3$ are associated with red, green and blue, respectively.
An inner point of $\Delta$ is expressed by the mixture of the three colors, which is illustrated in Figure \ref{color}.
Space-time diagrams of $3$-VCA and $3$-VFCA are shown in Figure \ref{stdiagram}.
They correspond to rule 6213370633633 for usual $3$-state $3$-neighbor CA,
where the rule numbers are the decimal number converted from 27 digits ternary number
of $3^{27}$ rules.
\begin{figure}[t]
\begin{center}
\scalebox{0.5}{\includegraphics{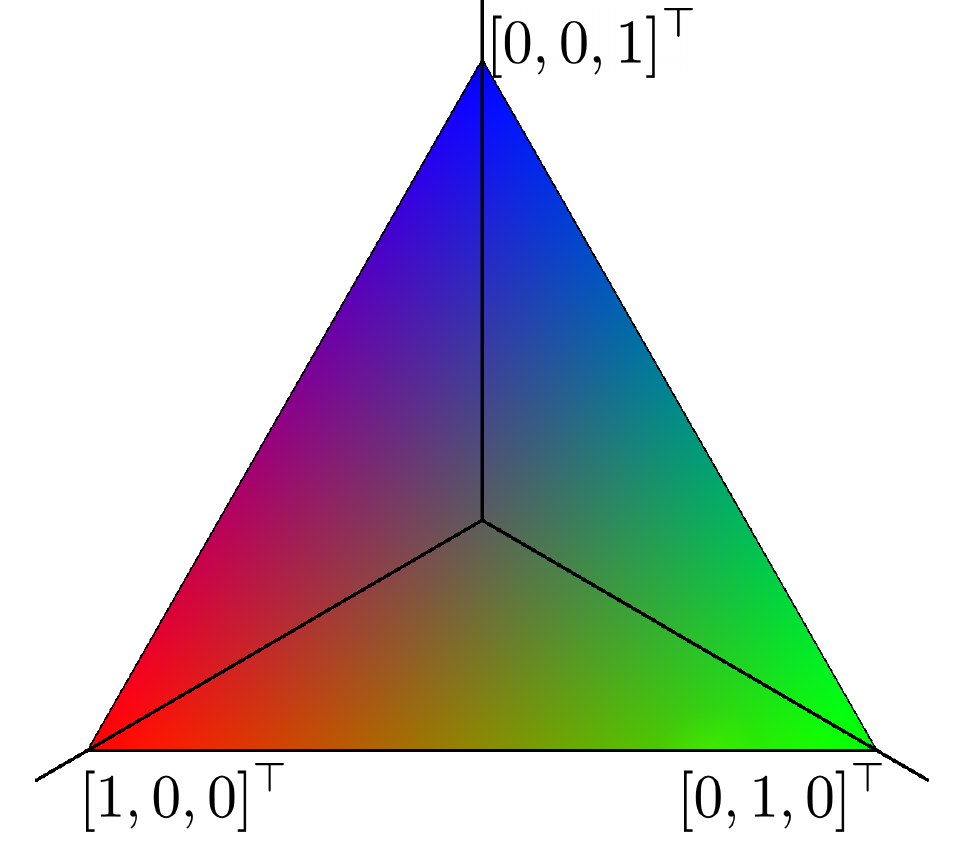}}
\end{center}
\caption{The color distribution on $\Delta$.}
\label{color}
\end{figure}
\begin{figure}[t]
\begin{center}
	\begin{minipage}{50mm}
	\begin{center}
	\scalebox{0.55}{\includegraphics{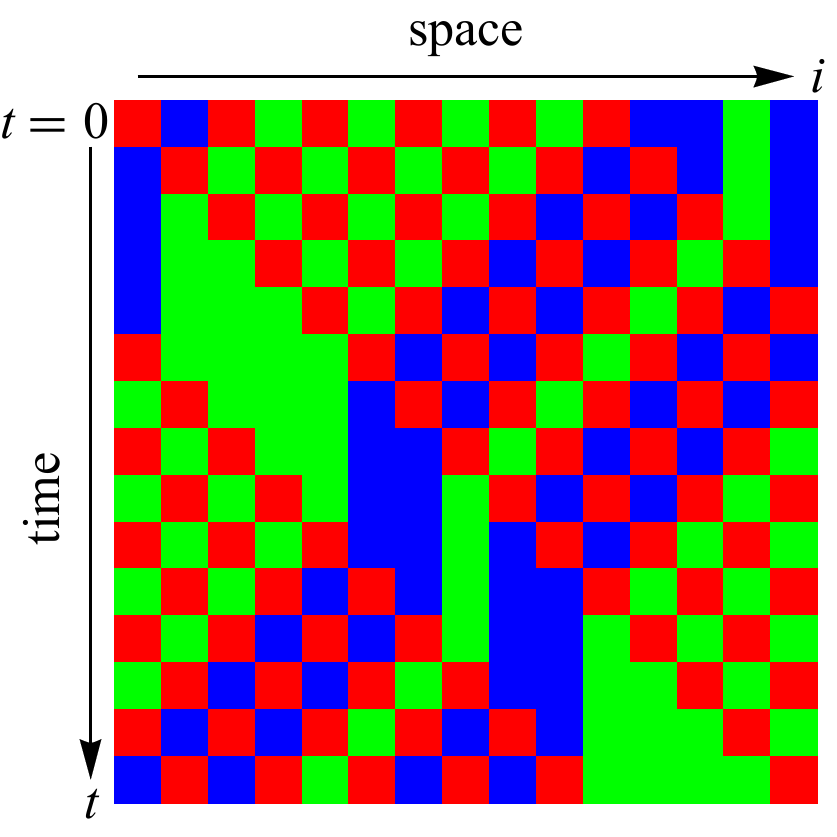}}
	\end{center}
	\end{minipage}
\hspace{3mm}
	\begin{minipage}{50mm}
	\begin{center}
	\scalebox{0.55}{\includegraphics{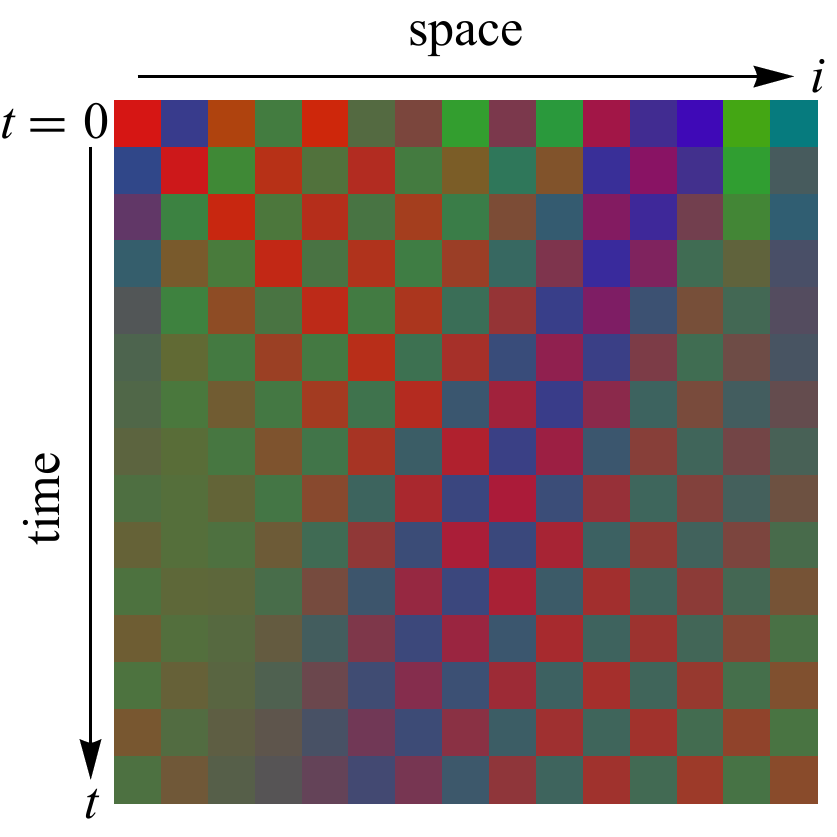}}
	\end{center}
	\end{minipage}
\end{center}
\caption{Space-time diagrams of $3$-VCA (left) and $3$-VFCA (right).}
\label{stdiagram}
\end{figure}
\section{Additive conserved quantities of $\lowercase{q}$-VFCA}
For $q$-VCA or $q$-VFCA,
let $\bm{x}^{t}_{i} = ( [\bm{x}^{t}_{i}]_{1}, [\bm{x}^{t}_{i}]_{2}, \dots, [\bm{x}^{t}_{i}]_{q} )^{\top}$ 
denote the state of the cell $i \in \mathbb{Z}$ at the time $t \in \mathbb{Z}_{\geq0}$.
In this section, we consider periodic CA with the positive integer period $L$, i.e.,
we assume $\bm{x}^{t}_{i} = \bm{x}^{t}_{i+L}$ for $i \in \mathbb{Z}$ and $t \in \mathbb{Z}_{\geq0}$.
\par
For a function $F: (\mathbb{R}^{q})^{p} \to \mathbb{R}$, the sum of the form
\begin{align*}
\Phi^t = \sum_{i=1}^{L} F( \bm{x}^{t}_{i}, \bm{x}^{t}_{i+1}, \dots, \bm{x}^{t}_{i+p-1} ) 
\end{align*}
is called an additive quantity of $q$-VCA or $q$-VFCA.
An additive quantity $\Phi^{t}$ is called an additive conserved quantity if it satisfies $\Phi^{t+1} = \Phi^{t}$ 
for any $t \in \mathbb{Z}_{\geq 0}$ and any initial configuration $(\bm{x}^{0}_{1}, \bm{x}^{0}_{2}, \dots, \bm{x}^{0}_{L})$. 
An example of additive quantities for $q$-VCA is the number of the cells in a period 
whose states are $\bm{e}_{k}$, which is counted as 
\begin{align*}
\nu^{t} (k) = \sum_{i=1}^{L} [\bm{x}^{t}_{i}]_{k}.
\end{align*}
These additive quantities can be extended for $q$-VFCA.
A rule of $q$-VCA or $q$-VFCA is called a $k$-number-conserving rule if $\nu^{t} (k)$ is the additive conserved quantity.
In particular, a rule is called a complete number-conserving rule if $\nu^{t} (k)$ is the additive conserved quantity
for all $k=1,2,\dots,q$.
\par
The vector representation of $q$-state CA or $q$-state fuzzy CA enables us to enumerate the rules
that admit the various kinds of additive conserved quantities.
We now demonstrate the enumeration of the number-conserving rules of $3$-VFCA with three neighboring cells.
Note that the computation below shows that a rule of $3$-VFCA is $k$-number-conserving
if and only if the corresponding $3$-VCA is $k$-number-conserving.
\par
Recall that the local rule of $3$-VFCA is of the form
\begin{align}
f(\bm{x},\bm{y},\bm{z}) = \left( 
\begin{array}{c}
 \displaystyle \sum_{j,k,\ell=1}^{3} a_{jk\ell} x_{j} y_{k} z_{\ell} \\
 \displaystyle \sum_{j,k,\ell=1}^{3} b_{jk\ell} x_{j} y_{k} z_{\ell} \\
 \displaystyle \sum_{j,k,\ell=1}^{3} c_{jk\ell} x_{j} y_{k} z_{\ell}
\end{array}
\right), \label{eq-fuzzylocal}
\end{align}
where $(a_{jk\ell},b_{jk\ell},c_{jk\ell})^{\top} \in \{\bm{e}_1,\bm{e}_2,\bm{e}_3\}$.
If the rule of $3$-VFCA defined by the local rule $f$ is 1-number-conserving, we have
\begin{equation} \label{eq-consv1}
\sum_{i=1}^{L} [\bm{x}^{t}_{i}]_{1} = \sum_{i=1}^{L} [\bm{x}^{t+1}_{i}]_{1} = 
\sum_{i=1}^{L} \sum_{j,k,\ell=1}^{3} a_{jk\ell} [\bm{x}^{t}_{i-1}]_{j} [\bm{x}^{t}_{i}]_{k} [\bm{x}^{t}_{i+1}]_{\ell}.
\end{equation}
Using the identities
\begin{align*}
[\bm{x}^{t}_{i}]_{1} + [\bm{x}^{t}_{i}]_{2} + [\bm{x}^{t}_{i}]_{3} = 1, \quad i=1,2,\dots,L,
\end{align*}
and the periodic conditions
\begin{align*}
&\sum_{i=1}^{L} [\bm{x}^{t}_{i-1}]_{k} = \sum_{i=1}^{L} [\bm{x}^{t}_{i}]_{k} = \sum_{i=1}^{L} [\bm{x}^{t}_{i+1}]_{k}, \quad k=1,2,3, \\
&\sum_{i=1}^{L} [\bm{x}^{t}_{i-1}]_{k} [\bm{x}^{t}_{i}]_{\ell} = \sum_{i=1}^{L} [\bm{x}^{t}_{i}]_{k} [\bm{x}^{t}_{i+1}]_{\ell} , \quad k,\ell=1,2,3,
\end{align*}
the most right-hand side of \eqref{eq-consv1} is computed as 
\begin{align*}
&a_{333} S(0,0,0) \\
&+ (a_{133} + a_{313} + a_{331} - 3a_{333}) S(0,1,0) \\
&+ (a_{233} + a_{323} + a_{332} - 3a_{333}) S(0,2,0) \\
&+ (a_{113} - a_{133} + a_{311} - 2a_{313} - a_{331} + 2a_{333}) S(1,1,0) \\
&+ (a_{213} - a_{233} - a_{313} + a_{321} - a_{323} - a_{331} + 2a_{333}) S(2,1,0) \\
&+ (a_{123} - a_{133} + a_{312} - a_{313} -a_{323} - a_{332} + 2a_{333}) S(1,2,0) \\
&+ (a_{223} - a_{233} + a_{322} - 2a_{323} - a_{332} + 2a_{333}) S(2,2,0) \\
&+ (a_{131} - a_{133} -a_{331} + a_{333}) S(1,0,1) \\
&+ (a_{231} - a_{233} - a_{331} + a_{333}) S(2,0,1) \\
&+ (a_{132} - a_{133} - a_{332} + a_{333}) S(1,0,2) \\
&+ (a_{232} - a_{233} - a_{332} + a_{333}) S (2,0,2) \\
&+ (a_{111} - a_{113} - a_{131} + a_{133} - a_{311} + a_{313} + a_{331} - a_{333})  S(1,1,1) \\
&+ (a_{211} - a_{213} - a_{231} + a_{233} - a_{311} + a_{313} + a_{331} - a_{333})  S(2,1,1) \\
&+ (a_{121} - a_{123} - a_{131} + a_{133} - a_{321} + a_{323} + a_{331} - a_{333})  S(1,2,1) \\
&+ (a_{221} - a_{223} - a_{231} + a_{233} - a_{321} + a_{323} + a_{331} - a_{333}) S(2,2,1) \\
&+ (a_{112} - a_{113} - a_{132} + a_{133} - a_{312} + a_{313} + a_{332} - a_{333}) S(1,1,2) \\
&+ (a_{212} - a_{213} - a_{232} + a_{233} - a_{312} + a_{313} + a_{332} - a_{333})  S(2,2,1) \\
&+ (a_{122} - a_{123} - a_{132} + a_{133} - a_{322} + a_{323} + a_{332} - a_{333}) S(1,2,2) \\
&+ (a_{222} - a_{223} - a_{232} + a_{233} - a_{322} + a_{323} +a_{332} - a_{333}) S(2,2,2).
\end{align*}
Here, 
$S(j,k,\ell) = \sum [\bm{x}^{t}_{i-1}]_{j} [\bm{x}^{t}_{i}]_{k} [\bm{x}^{t}_{i+1}]_{\ell}$ for $j,k,\ell = 0,1,2,$ 
$[\bm{x}^{t}_{i}]_{0} = 1$,
and
each summation is taken over $i=1,2,\dots,L$.
Regarding $[\bm{x}^{t}_{i}]_{1}, [\bm{x}^{t}_{i}]_{2}, i=1,2,\dots,L,$ as independent indeterminates and comparing the coefficients in \eqref{eq-consv1},
we have 19 equations.
Solving them under the condition $a_{jk\ell} \in \{0,1\}$, we obtain nine solutions for $(a_{111},a_{112},\ldots,a_{333})$,
see Appendix A.
If $a_{jk\ell} = 1$, then we have $b_{jk\ell} = c_{jk\ell} = 0$.
On the other hand, if $a_{jk\ell} = 0$, then we have two possibilities: $b_{jk\ell}=1, c_{jk\ell} = 0$ and $b_{jk\ell} = 0, c_{jk\ell} = 1$.
The number of the $1$-number-conserving $3$-VFCA rules is computed as $9 \times 2^{18}$.
In the same way, we can obtain the conditions for $b_{jk\ell}$ (resp.~$c_{jk\ell}$) for 
the $2$-number conserving (resp.~the $3$-number-conserving) rules, which are shown in Appendix B.
\par
To find complete number-conserving rules,
we seek local rules satisfying both the conditions for $a_{jk\ell}$ and for $b_{jk\ell}$.
Note that if $\nu^{t}(1)$ and $\nu^{t}(2)$ are additive conserved quantities, so is $\nu^{t}(3)$.
A pair of $(a_{111},a_{112},\dots,a_{333})$ and $(b_{111},b_{112},\dots,b_{333})$ gives the local rule of a $3$-VFCA
if and only if $a_{jk\ell} + b_{jk\ell} \leq 1$ for all $1 \leq j,k,\ell \leq 3$.
Hence, checking $9 \times 9$ pairs of $(a_{111},a_{112},\dots,a_{333})$ and $(b_{111},b_{112},\dots,b_{333})$,
we obtain 15 complete number-conserving rules, shown in Appendix C.
\section{Concluding remarks}
In this paper, we first introduce the vector representation of $q$-state $n$-neighbor CA called $q$-VCA.
The $q$ states are assigned to the standard basis vectors $\bm{e}_{1},\bm{e}_{2},\dots,\bm{e}_{q}$ of $\mathbb{R}^{q}$ and
the local rule $h$ can be expressed by a tuple of $q$ polynomials that are homogeneous of degree $n$.
Then, we consider CA with the states in the convex hull of $\bm{e}_{1},\bm{e}_{2},\dots,\bm{e}_{q}$,
where we naturally expand the vector-valued map $h$ via the polynomial expression.
We call them $q$-VFCA.
If $q=2$, the local rule for $2$-VFCA is equivalent to that for elementary fuzzy CA
obtained from the disjunctive normal form in boolean operations.
We also explain how to enumerate the number-conserving rules of these vector-valued CA.
If we focus on the additive quantity 
$F(\bm{x}_{i}^{t}) = \sum_{k=1}^{q} (k-1)[\bm{x}_{i}^{t}]_{k}$,
we can discuss
the usual number-conserving rules~\cite{Fuks}, which conserve the number of particles, that is, the sum of the values.
It is a future work to investigate the properties of $q$-VFCA, 
such as the asymptotic behavior.
\appendix
\section{Solving equations for $1$-number-conserving rules}
Comparing the coefficients of \eqref{eq-consv1}, we have the following 19 equations.
\begin{align}
&a_{333} = 0, \label{eq-a1} \\
&a_{133} + a_{313} + a_{331} - 3a_{333} = 1, \label{eq-a2} \\
&a_{233} + a_{323} + a_{332} - 3a_{333} = 0, \label{eq-a3} \\
&a_{113} - a_{133} + a_{311} - 2a_{313} - a_{331} + 2a_{333} = 0, \label{eq-a4} \\
&a_{213} - a_{233} - a_{313} + a_{321} - a_{323} - a_{331} + 2a_{333} = 0, \label{eq-a5} \\
&a_{123} - a_{133} + a_{312} - a_{313} -a_{323} - a_{332} + 2a_{333} = 0, \label{eq-a6} \\
&a_{223} - a_{233} + a_{322} - 2a_{323} - a_{332} + 2a_{333} = 0, \label{eq-a7} \\
&a_{131} - a_{133} -a_{331} + a_{333} = 0, \label{eq-a8} \\
&a_{231} - a_{233} - a_{331} + a_{333} = 0, \label{eq-a9} \\
&a_{132} - a_{133} - a_{332} + a_{333} = 0, \label{eq-a10} \\
&a_{232} - a_{233} - a_{332} + a_{333} = 0, \label{eq-a11} \\
&a_{111} - a_{113} - a_{131} + a_{133} - a_{311} + a_{313} + a_{331} - a_{333} = 0, \label{eq-a12} \\
&a_{211} - a_{213} - a_{231} + a_{233} - a_{311} + a_{313} + a_{331} - a_{333} = 0, \label{eq-a13} \\
&a_{121} - a_{123} - a_{131} + a_{133} - a_{321} + a_{323} + a_{331} - a_{333} = 0, \label{eq-a14} \\
&a_{221} - a_{223} - a_{231} + a_{233} - a_{321} + a_{323} + a_{331} - a_{333} = 0, \label{eq-a15} \\
&a_{112} - a_{113} - a_{132} + a_{133} - a_{312} + a_{313} + a_{332} - a_{333} = 0, \label{eq-a16} \\
&a_{212} - a_{213} - a_{232} + a_{233} - a_{312} + a_{313} + a_{332} - a_{333} = 0, \label{eq-a17} \\
&a_{122} - a_{123} - a_{132} + a_{133} - a_{322} + a_{323} + a_{332} - a_{333} = 0, \label{eq-a18} \\
&a_{222} - a_{223} - a_{232} + a_{233} - a_{322} + a_{323} +a_{332} - a_{333} = 0. \label{eq-a19} 
\end{align}
From \eqref{eq-a1}, \eqref{eq-a3}, \eqref{eq-a7}, \eqref{eq-a11} and \eqref{eq-a19},
we have 
\begin{align*}
a_{222} = a_{223} = a_{232} = a_{233} = a_{322} = a_{323} = a_{332} = a_{333} = 0.
\end{align*}
We consider three cases for \eqref{eq-a2}.
\begin{enumerate}
\item $a_{133} = 1$ and $a_{313} = a_{331} = 0$.\\
We have
\begin{align*}
a_{111} = a_{131} = a_{132} = 1, \quad a_{213} = a_{221} = a_{231} = a_{321} = 0
\end{align*}
from \eqref{eq-a5}, \eqref{eq-a8}, \eqref{eq-a9}, \eqref{eq-a10}, \eqref{eq-a12} and \eqref{eq-a15}.
Now, equations \eqref{eq-a4} and \eqref{eq-a6} turn to
\begin{align*}
a_{113} + a_{311} = 1, \quad a_{123} + a_{312} = 1.
\end{align*}
Hence, we have four possibilities for $(a_{113}, a_{311}, a_{123}, a_{312})$.
If  $a_{113}=1, a_{311}=0, a_{123}=0$ and $a_{312} = 1$, equation \eqref{eq-a16} implies $a_{112} = 2$, leading to a contradiction.
Thus, we obtain the following three cases, where $a_{112}, a_{121}, a_{122}, a_{211}$ and $a_{221}$ are
uniquely determined by \eqref{eq-a13}, \eqref{eq-a14}, \eqref{eq-a16}, \eqref{eq-a17} and \eqref{eq-a18}.
\begin{align*}
&(a_{113}, a_{311}, a_{123}, a_{312}, a_{112}, a_{121}, a_{122}, a_{211}, a_{221}) \\
=& \begin{cases}
(1,0,1,0,1,1,1,0,0), \\
(0,1,1,0,0,1,1,1,0), \\
(0,1,0,1,1,0,0,1,1).
\end{cases}
\end{align*}
\item $a_{313} = 1$ and $a_{133} = a_{331} = 0$.\\
We have
\begin{align*}
a_{111} = a_{113} = a_{311} = 1, \quad a_{131} = a_{231} = a_{132} = 0
\end{align*}
from \eqref{eq-a4}, \eqref{eq-a8}, \eqref{eq-a9}, \eqref{eq-a10} and \eqref{eq-a12}.
Equations \eqref{eq-a5} and \eqref{eq-a6} turn to
\begin{align*}
a_{213} + a_{321} = 1, \quad a_{123} + a_{312} = 1.
\end{align*}
Again, we have four possibilities for $(a_{213}, a_{321}, a_{123}, a_{312})$.
If  $a_{213}=0, a_{321}=1, a_{123}=1$ and $a_{312} = 0$, equation \eqref{eq-a18} implies $a_{112} = -1$, leading to a contradiction.
Thus, we obtain the following three cases.
\begin{align*}
&(a_{213}, a_{321}, a_{123}, a_{312}, a_{112}, a_{121}, a_{122}, a_{211}, a_{212}, a_{221}) \\
=& \begin{cases}
(1,0,1,0,0,1,1,1,0,0), \\
(1,0,0,1,1,0,0,1,1,0), \\
(0,1,0,1,1,1,0,0,0,1).
\end{cases}
\end{align*}
\item $a_{331} = 1$ and $a_{133} = a_{313} = 0$.\\
In the same way, we have
\begin{align*}
a_{111} = a_{131} = a_{231} = 1, \quad a_{122} = a_{123} = a_{132} = a_{312} = 0
\end{align*}
from \eqref{eq-a6}, \eqref{eq-a8}, \eqref{eq-a9}, \eqref{eq-a10}, \eqref{eq-a12} and \eqref{eq-a18}, and
\begin{align*}
&(a_{113}, a_{311}, a_{213}, a_{321}, a_{112}, a_{121}, a_{211}, a_{212}, a_{221}) \\
= &\begin{cases}
(1,0,1,0,1,0,1,1,0), \\
(1,0,0,1,1,1,0,0,1), \\
(0,1,0,1,0,1,1,0,1).
\end{cases}
\end{align*}
from other equations.
\end{enumerate}
Hence, we have the following nine solutions.
\begin{align} \label{solA}
\notag &(a_{111},a_{112},\ldots,a_{333})=\\
&\qquad\begin{cases}
(1, 1, 1, 1, 1, 1, 1, 1, 1, 0, 0, 0, 0, 0, 0, 0, 0, 0, 0, 0, 0, 0, 0, 0, 0, 0, 0) ,\\
(1, 1, 1, 1, 0, 0, 1, 0, 0, 0, 0, 0, 1, 0, 0, 1, 0, 0, 0, 0, 0, 1, 0, 0, 1, 0, 0) ,\\
(1, 1, 1, 1, 0, 0, 0, 0, 0, 0, 0, 0, 1, 0, 0, 0, 0, 0, 1, 1, 1, 1, 0, 0, 0, 0, 0) ,\\
(1, 1, 1, 0, 0, 0, 1, 0, 0, 1, 1, 1, 0, 0, 0, 1, 0, 0, 0, 0, 0, 0, 0, 0, 1, 0, 0) ,\\
(1, 1, 1, 0, 0, 0, 0, 0, 0, 1, 1, 1, 0, 0, 0, 0, 0, 0, 1, 1, 1, 0, 0, 0, 0, 0, 0) ,\\
(1, 1, 0, 0, 0, 0, 1, 1, 1, 1, 1, 0, 0, 0, 0, 0, 0, 0, 1, 1, 0, 0, 0, 0, 0, 0, 0) ,\\
(1, 0, 1, 1, 1, 1, 0, 0, 0, 1, 0, 1, 0, 0, 0, 0, 0, 0, 1, 0, 1, 0, 0, 0, 0, 0, 0) ,\\
(1, 0, 0, 1, 1, 1, 1, 1, 1, 1, 0, 0, 0, 0, 0, 0, 0, 0, 1, 0, 0, 0, 0, 0, 0, 0, 0) ,\\
(1, 0, 0, 1, 0, 0, 1, 0, 0, 1, 0, 0, 1, 0, 0, 1, 0, 0, 1, 0, 0, 1, 0, 0, 1, 0, 0).
\end{cases} 
\end{align}
\section{Solutions for other number-conserving rules}
A rule of $3$-VFCA is $2$-number-conserving if and only if the coefficients of the local rule \eqref{eq-fuzzylocal} satisfies
\begin{align} \label{solB}
\notag &(b_{111},b_{112},\ldots,b_{333})=\\ 
&\qquad\begin{cases}
(0, 1, 0, 0, 1, 0, 0, 1, 0, 0, 1, 0, 0, 1, 0, 0, 1, 0, 0, 1, 0, 0, 1, 0, 0, 1, 0) ,\\
(0, 1, 0, 0, 0, 0, 0, 1, 0, 0, 1, 0, 1, 1, 1, 0, 1, 0, 0, 1, 0, 0, 0, 0, 0, 1, 0) ,\\
(0, 1, 0, 0, 0, 0, 0, 0, 0, 0, 1, 0, 1, 1, 1, 0, 0, 0, 0, 1, 0, 1, 1, 1, 0, 0, 0) ,\\
(0, 0, 0, 1, 1, 1, 0, 1, 0, 0, 0, 0, 1, 1, 1, 0, 1, 0, 0, 0, 0, 0, 0, 0, 0, 1, 0) ,\\
(0, 0, 0, 1, 1, 1, 0, 0, 0, 0, 0, 0, 1, 1, 1, 0, 0, 0, 0, 0, 0, 1, 1, 1, 0, 0, 0) ,\\
(0, 0, 0, 1, 1, 0, 0, 0, 0, 0, 0, 0, 1, 1, 0, 1, 1, 1, 0, 0, 0, 1, 1, 0, 0, 0, 0) ,\\
(0, 0, 0, 0, 1, 1, 0, 0, 0, 1, 1, 1, 0, 1, 1, 0, 0, 0, 0, 0, 0, 0, 1, 1, 0, 0, 0) ,\\
(0, 0, 0, 0, 1, 0, 0, 0, 0, 1, 1, 1, 0, 1, 0, 1, 1, 1, 0, 0, 0, 0, 1, 0, 0, 0, 0) ,\\
(0, 0, 0, 0, 0, 0, 0, 0, 0, 1, 1, 1, 1, 1, 1, 1, 1, 1, 0, 0, 0, 0, 0, 0, 0, 0, 0).
\end{cases}
\end{align}
A rule of $3$-VFCA is $3$-number-conserving if and only if the coefficients of the local rule \eqref{eq-fuzzylocal} satisfies
\begin{align*}
&(c_{111},c_{112},\ldots,c_{333})=\\ 
&\qquad\begin{cases}
(0, 0, 1, 0, 0, 1, 0, 0, 1, 0, 0, 1, 0, 0, 1, 0, 0, 1, 0, 0, 1, 0, 0, 1, 0, 0, 1) ,\\
(0, 0, 1, 0, 0, 1, 0, 0, 0, 0, 0, 1, 0, 0, 1, 0, 0, 0, 0, 0, 1, 0, 0, 1, 1, 1, 1) ,\\
(0, 0, 1, 0, 0, 0, 0, 0, 0, 0, 0, 1, 0, 0, 0, 1, 1, 1, 0, 0, 1, 0, 0, 0, 1, 1, 1) ,\\
(0, 0, 0, 0, 0, 1, 1, 1, 1, 0, 0, 0, 0, 0, 1, 0, 0, 0, 0, 0, 0, 0, 0, 1, 1, 1, 1) ,\\
(0, 0, 0, 0, 0, 0, 1, 1, 1, 0, 0, 0, 0, 0, 0, 1, 1, 1, 0, 0, 0, 0, 0, 0, 1, 1, 1) ,\\
(0, 0, 0, 0, 0, 0, 1, 0, 1, 0, 0, 0, 0, 0, 0, 1, 0, 1, 0, 0, 0, 1, 1, 1, 1, 0, 1) ,\\
(0, 0, 0, 0, 0, 0, 0, 1, 1, 0, 0, 0, 0, 0, 0, 0, 1, 1, 1, 1, 1, 0, 0, 0, 0, 1, 1) ,\\
(0, 0, 0, 0, 0, 0, 0, 0, 1, 0, 0, 0, 0, 0, 0, 0, 0, 1, 1, 1, 1, 1, 1, 1, 0, 0, 1) ,\\
(0, 0, 0, 0, 0, 0, 0, 0, 0, 0, 0, 0, 0, 0, 0, 0, 0, 0, 1, 1, 1, 1, 1, 1, 1, 1, 1).
\end{cases}
\end{align*}
\section{Complete number-conserving rules}
Table \ref{tab-consvcomp} shows the rule numbers and the local rules of complete number-conserving $3$-VFCA.
The rule numbers of $3$-VCA or $3$-VFCA are determined as follows.
Let $h: \{\bm{e}_1,\bm{e}_2,\bm{e}_3\}^3 \to \{\bm{e}_1,\bm{e}_2,\bm{e}_3\}$ be the local rule of $3$-VCA.
Then, we consider the sequence 
\begin{align*}
h(\bm{e}_{1}, \bm{e}_{1}, \bm{e}_{1}) h(\bm{e}_{1}, \bm{e}_{1}, \bm{e}_{2}) \cdots h(\bm{e}_{3}, \bm{e}_{3}, \bm{e}_{3}).
\end{align*}
Replacing $\bm{e}_{1}$ with $2$, $\bm{e}_{2}$ with $1$, and $\bm{e}_{3}$ with $0$,
we have 27 digits ternary number.
We can compute the rule number in the decimal number by converting this ternary number.
The rule numbers of $3$-VFCA are the same as those of the corresponding $3$-VCA.
We also express the rule numbers in the $27$-adic numbers, which are expressed by 
$0,1,\dots,9,\rm{a},\rm{b},\dots,\rm{q}$,
 since they seem to capture the combinatorial characteristics of complete number-conserving rules.
The pair $(\mu,\nu)$ in the table means that
the rule is determined by the $\mu$th solution in \eqref{solA} and the $\nu$th solution in \eqref{solB}.
\begin{table}
\begin{center}
\begin{tabular}{|lll|}\hline
decimal rule number  &  $27$-adic rule number  &  pair\\
\multicolumn{3}{|l|}{ \qquad local rule} \\ \hline\hline
$7479532539765$  &  $\rm{qd0qd0qd0}$  &  $(5,5)$\\
\multicolumn{3}{|l|}{ \qquad $(y_1,y_2,y_3)^{\top}$} \\ \hline
$6159136430181$  &  $\rm{\ell \ell \ell \ell \ell \ell \ell \ell \ell}$  &  $(9,1)$\\
\multicolumn{3}{|l|}{ \qquad $(z_1,z_2,z_3)^{\top}$} \\
$7625403764901$  &  $\rm{qqqddd000}$  &  $(1,9)$\\
\multicolumn{3}{|l|}{ \qquad $(x_1,x_2,x_3)^{\top}$} \\ \hline
$6768185473053$  &  $\rm{nq0nd0nd0}$  &  $(7,3)$\\
\multicolumn{3}{|l|}{ \qquad $(x_1y_2+y_1z_1+y_1z_3,x_2y_2+x_3y_2+y_1z_2,y_3)^{\top}$} \\
$6924717700245$  &  $\rm{odqod0od0}$  &  $(6,5)$\\
\multicolumn{3}{|l|}{ \qquad $(x_1y_3+y_1z_1+y_1z_2,y_2,x_2y_3+x_3y_3+y_1z_3)^{\top}$} \\
$7469071910973$  &  $\rm{qc0qcdqc0}$  &  $(5,6)$\\
\multicolumn{3}{|l|}{ \qquad $(y_1,x_2y_3+y_2z_1+y_2z_2,x_1y_3+x_3y_3+y_2z_3)^{\top}$} \\
$7480694859933$  &  $\rm{qd3qd3q03}$  &  $(5,4)$\\
\multicolumn{3}{|l|}{ \qquad $(y_1,x_1y_2+x_2y_2+y_3z_2,x_3y_2+y_3z_1+y_3z_3)^{\top}$} \\
$7486506443925$  &  $\rm{qdiqdi0di}$  &  $(4,5)$\\
\multicolumn{3}{|l|}{ \qquad $(x_1y_1+x_2y_1+y_3z_1,y_2,x_3y_1+y_3z_2+y_3z_3)^{\top}$} \\
$7573493966013$  &  $\rm{qm0dm0qm0}$  &  $(3,7)$\\
\multicolumn{3}{|l|}{ \qquad $(x_1y_1+x_3y_1+y_2z_1,x_2y_1+y_2z_2+y_2z_3,y_3)^{\top}$} \\ \hline
$6213370633533$  &  $\rm{\ell qq \ell d0 \ell d0}$  &  $(8,3)$\\
\multicolumn{3}{|l|}{ \qquad $(x_1y_2+x_1y_3+y_1z_1,x_2y_2+x_3y_2+y_1z_2,x_2y_3+x_3y_3+y_1z_3)^{\top}$} \\
$6769347793221$  &  $\rm{nq3nd3n03}$  &  $(7,2)$\\
\multicolumn{3}{|l|}{ \qquad $(x_1y_2+y_1z_1+y_1z_3,x_2y_2+y_1z_2+y_3z_2,x_3y_2+y_3z_1+y_3z_3)^{\top}$} \\
$6914257071453$  &  $\rm{ocqocdoc0}$  &  $(6,6)$\\ 
\multicolumn{3}{|l|}{ \qquad $(x_1y_3+y_1z_1+y_1z_2,x_2y_3+y_2z_1+y_2z_2,x_3y_3+y_1z_3+y_2z_3)^{\top}$} \\ 
$7487668764093$  &  $\rm{qd \ell qd \ell 00 \ell}$  &  $(4,4)$\\
\multicolumn{3}{|l|}{ \qquad $(x_1y_1+x_2y_1+y_3z_1,x_1y_2+x_2y_2+y_3z_2,x_3y_1+x_3y_2+y_3z_3)^{\top}$} \\
$7563033337221$  &  $\rm{q \ell 0d \ell dq \ell 0}$  &  $(3,8)$\\
\multicolumn{3}{|l|}{ \qquad $(x_1y_1+x_3y_1+y_2z_1,x_2y_1+x_2y_3+y_2z_2,x_1y_3+x_3y_3+y_2z_3)^{\top}$} \\
$7580467870173$  &  $\rm{qmidmi0mi}$  &  $(2,7)$\\
\multicolumn{3}{|l|}{ \qquad $(x_1y_1+y_2z_1+y_3z_1,x_2y_1+y_2z_2+y_2z_3,x_3y_1+y_3z_2+y_3z_3)^{\top}$} \\ \hline
\end{tabular}
\end{center}
\caption{Complete number-conserving rules of $3$-VFCA.} \label{tab-consvcomp}
\end{table}

\end{document}